\documentclass[12pt,prb]{revtex4}

\usepackage{amsfonts}
\usepackage{amsthm}
\usepackage{amsmath}
\usepackage{amssymb}
\usepackage{amscd}
\usepackage{eucal}
\usepackage[dvips]{graphicx}

\newtheorem{definition}{Definition}
\newtheorem{proposition}{Proposition}

\newcommand{\Imag}{{\mbox{Im}\,}}

\newcommand{\Sg}{\mbox{sgn}}

\newcommand{\rmi}{{\rm i}}
\newcommand{\rmd}{{\rm d}}

\newcommand{\C}{\mathbb{C}}
\newcommand{\R}{\mathbb{R}}

\newcommand{\HH}{\mathbb{H}}

\newcommand{\SSS}{\mathbb{S}}
\newcommand{\cB}{{\cal B}}

\newcommand{\cP}{{\cal P}}
\newcommand{\cU}{{\cal U}}

\newcommand{\bk}{{\mathbf k}}
\newcommand{\bx}{{\mathbf x}}

\newcommand{\br}{{\mathbf r}}
\newcommand{\bo}{\mbox{\boldmath$\omega$}}
\newcommand{\bos}{\mbox{\boldmath$\scriptstyle\omega$}}

\newcommand{\beq}{\begin{equation}}
\newcommand{\eeq}{\end{equation}}
\newcommand{\Cs}{{\mbox{const.}}}

\begin{document}

\title{Hyperbolic geometrical optics: Hyperbolic glass}
\author{Enrico De Micheli}
\affiliation{IBF -- Consiglio Nazionale delle Ricerche, Via De Marini, 6 - 16149 Genova, Italy.}
\author{Irene Scorza}
\affiliation{Dipartimento di Matematica - Universit\`a di Genova \\ Via Dodecaneso, 35 - 16146 Genova, Italy.}
\author{Giovanni Alberto Viano}
\affiliation{Dipartimento di Fisica - Universit\`a di Genova \\
Istituto Nazionale di Fisica Nucleare - sez. di Genova\\ Via Dodecaneso, 33 - 16146 Genova, Italy.}

\begin{abstract}
We study the geometrical optics generated by a refractive index
of the form $n(x,y)=1/y$ $(y>0)$, where $y$ is the coordinate of the
vertical axis in an orthogonal reference frame in $\R^2$. We thus obtain
what we call ``hyperbolic geometrical optics'' since the ray trajectories
are geodesics in the Poincar\'e-Lobachevsky half--plane $\HH^2$. Then
we prove that the constant phase surface are horocycles and obtain the \emph{horocyclic waves},
which are closely related to the classical Poisson kernel and are the analogs of the Euclidean plane waves.
By studying the transport equation in the Beltrami pseudosphere, we prove
(i) the conservation of the flow in the entire strip $0<y\leqslant 1$
in $\HH^2$, which is the limited region of physical interest where the ray
trajectories lie; (ii) the nonuniform distribution of the density of trajectories:
the rays are indeed focused toward the horizontal $x$ axis, which is
the boundary of $\HH^2$. Finally the process of ray focusing and defocusing
is analyzed in detail by means of the sine--Gordon equation.
\end{abstract}

\maketitle

\newpage

\section{Introduction}
\label{se:introduction}
It is well known that the geometrical optics approximation of the wave equation
is related to the asymptotic form of the
integral representation of the field (if such exists), which is an exact solution of the
wave problem. Suppose, for instance, that the field in a uniform medium can be
written as an expansion in plane waves; the evaluation of this integral
by the stationary phase method yields an asymptotic series. Then the leading term
of this asymptotic expansion, which is composed by an amplitude and a phase,
can be extracted to yield the approximation.
The ray trajectories are the lines orthogonal to the constant phase surface and are
described by the eikonal equation; the amplitude satisfies the transport equation, whose
physical meaning is related to the conservation of the flow. In the simplest case of
uniform medium, whose refractive index $n$ is a real constant, the rays are straight lines
which are characterized by the following properties:
\begin{itemize}
\item[(i)] They are geodesics of the Euclidean space.
\item[(ii)] Phase and amplitude are real--valued functions.
\item[(iii)] They can be derived by the Fermat's principle.
\end{itemize}
Constrained by these properties the methods of geometrical optics
are rather limited and fail to explain several phenomena as, for instance,
the diffraction by a compact and opaque obstacle, that is the existence of
non--null field in the geometrical shadow which, for this reason, is usually
referred to as the classically (or geometrically) forbidden region.

In the decade 1950--1960 J. B. Keller \cite{Keller,Levy,Hansen} wrote several papers where he introduced
the so--called Geometrical Theory of Diffraction (GTD). The latter can be regarded
as an extension of geometrical optics, which accounts for diffraction by introducing
the diffracted rays in addition to the usual rays of geometrical optics. After these seminal
works there has been a steady flow of papers addressing various aspects of the theory.
On the one hand papers oriented to pure and applied electromagnetic theory, like radiation
and scattering of waves, antenna design, waveguide theory and so on \cite{Hansen}; on the other hand, a
highly theoretical and mathematically sophisticated theory of propagation of singularities
and diffraction of waves on manifolds \cite{Melrose}. In spite of these efforts and a wide literature
on these topics, not all the cases of interest have been studied. An example is
what we could call the ``hyperbolic geometrical optics'', that is the geometrical optics
generated by the rays in the specific case of a refractive index of the form $n(x,y)=1/y$ $(y>0)$,
where $y$ denotes a spatial coordinate, say vertical, in an appropriate orthogonal reference frame
in $\R^2$. As far as we know, this problem has never been treated, except for some very
marginal remarks (see, for instance, Ref. \onlinecite{Smirnov}), in spite of its intrinsic geometrical interest
and some possible applications to the physics of nonuniform optical fibers. It is precisely
the main purpose of the present paper to fill this gap.

Let us return to Keller's program of widening, from a geometrical viewpoint, the arena of
Euclidean geometrical optics. With this in mind we adopt, first of all, the Jacobi's form of the
principle of least action (instead of Fermat's), which concerns with the path of the system point
rather than with its time evolution \cite{Goldstein}.
More precisely, the Jacobi's principle (generally applied in mechanics)
can be formulated as follows: If there are no forces acting on the body, then the system travels along
the shortest path length in the configuration space. Here we assume a wide extension of
Jacobi's principle, which can be formulated as follows: the geodesics associated with the Riemannian
metric $n(x,y)\sqrt{\rmd x^2+\rmd y^2}$, i.e. the paths making the functional $\int n(x,y)\sqrt{\rmd x^2+\rmd y^2}$ stationary, are
nicknamed rays. In other words, in place of Fermat's principle which reads
$\delta\int_{P_0}^{P_1}\rmd t=0$, where $dt$ is the travel time measure, and $P_0$ and $P_1$
are prescribed starting and end points of the path, we write
\beq
\label{r1}
\delta\int_{P_0}^{P_1} n(x,y)\sqrt{\rmd x^2+\rmd y^2} = 0,
\eeq
or, equivalently,
\beq
\label{r2}
\delta\int_{x_0}^{x_1} F(x,y,y') \,\rmd x = 0,~~~~\left(F(x,y,y')=n(x,y)\sqrt{1+y'^2}\right),
\eeq
where $y'=\tan\alpha$, $\alpha$ being the angle that the tangent to the curve $y=y(x)$
forms with the $x$ axis.

The simplest realization of this Jacobi's principle consists in identifying
$n^2$ with the Riemann metric tensor $g_{ij}$, \emph{whenever this identification
is admissible}. This identification requires great caution, indeed;
the form $g_{ij}\rmd x^i \rmd x^j$ must be symmetric and positive definite, and this
poses a strict restriction. For instance, consider a refractive index (or, in mechanics, a potential)
of the following form: $n^2=1-V/E$, where $E$ is the energy of the incoming particle
and $V$ is the height of the potential, with $V>E$ as in the case of the tunnel effect.
In this situation the geometric interpretation of the trajectory as a real--valued geodesic
in a Riemannian manifold is no longer possible. The only chance remains to extend the
admissible values of the phase to imaginary and/or complex values and, consequently,
to speak of complex rays in the sense of Landau \cite{Landau}.

But let us return to the cases where this identification is admissible. As we already
mentioned, it is obviously possible in the case of a uniform nonabsorbing medium: in this
case we simply obtain a physical realization of Euclidean geometry. But it is also certainly
admissible when the refractive index is of the form introduced above, i.e., $n(x,y)=1/y$
($y>0$), where $y$ denotes the coordinate of the vertical axis in an orthogonal
reference frame in $\R^2$. In this case we are led to the Lobachevskian metric:
$\rmd s^2=(\rmd x^2+\rmd y^2)/y^2$. Then the rays are geodesics in the hyperbolic half--plane
(Poincar\'e half--plane): i.e., Euclidean half--circles with centers on the $x$ axis
(horizontal axis), or Euclidean straight lines normal to the $x$ axis.
Let us recall that the refractive index $n$ is defined as: $n=c/v_{\rm ph}$, where
$c$ is the light speed in vacuum, and $v_{\rm ph}$ is the phase velocity of
radiation of a specific frequency in a specific material.
Therefore $n \geqslant 1$, and in the case $n(x,y)=\frac{1}{y}$ only the strip $0<y\leqslant 1$
has physical interest; hence the actual rays will lie necessarily in this band.
Accordingly, hereafter, the only optical paths considered will be the Euclidean half--circles
with centers on the $x$ axis and radius $R$ bounded by $0<R\leqslant 1$.

The subsequent step in developing an optical geometry consists in finding the constant phase
surfaces and, accordingly, describing the analog of the Euclidean plane wave.
This problem will be solved in Sec. \ref{se:global}, studying some geometrical
properties of horocycles and introducing what we call \emph{horocyclic waves}, which
play in hyperbolic geometrical optics the same role as the plane waves do in the Euclidean one.
At this point we have the main ingredients needed for writing the geometrical approximation
of the wave function; what it is still missing is an analysis of the amplitude
and of the related flux density. This latter problem can be analyzed
at two different levels. First we prove that the flow of rays is conserved:
once a pointlike source is fixed, no ray will be absorbed or created.
This result will be proved in Sec. \ref{se:global}.
A more subtle question is the following: Is the flow of the ray trajectories
homogeneous or do the rays focus? This issue, besides its intrinsic geometrical interest,
could in our opinion be of some interest in possible applications to the
propagation in optical fibers with non--uniform refractive index \cite{Okoshi}.
This problem will be analyzed in detail in Sec. \ref{se:focusing}. First
we study the transport equation in the Beltrami pseudosphere, and prove that
the flow of ray trajectories is not homogeneous, but there is a focusing of
rays on the horizontal $x$ axis. Glancing to possible applications to propagation in
optical fibers this result suggests a conjecture indicating a strong ray focusing
along the fiber axis, when the refractive index profile
in the fiber is of hyperbolic type, instead of paraboliclike, as is customary.
Next, this problem will be reconsidered by
studying the variation of the angle that the tangent to the meridian
of the Beltrami pseudosphere makes with the rotation axis of this surface,
which can be indeed represented as a surface of revolution generated by a curve in $\R^3$.
This leads to the sine--Gordon equation and provides a more precise description of the ray
focusing and defocusing processes.
This analysis is necessarily local, since the problem is worked out inside each horocycle;
at the end of Sec. \ref{se:focusing}, we show how to pass from a local description
of the flow inside each horocycle to a global one.

Finally, in the Appendix, the geometric and algebraic ingredients which occur in
Secs. \ref{se:global} and \ref{se:focusing} will be given. This appendix is
split in three parts: the first part is devoted to the various models of
hyperbolic geometry and to the conformal maps which allows the transformation
between them; in the second part we study the group $SU(1,1)$, which
acts transitively on the non--Euclidean disk, and prove some relationships connecting
the spherical functions to the horocyclic waves; the last part is devoted to the
Beltrami pseudosphere.

\section{The flow in the strip $\mathbf{0<y\leqslant 1}$}
\label{se:global}

\subsection{Variational minimization of the Jacobi's functional and the rays in hyperbolic geometrical optics}
\label{subse:variational}

Let us consider the upper half--plane model of the hyperbolic
two--dimensional space $\HH^2$: i.e., $U=\{z=x+iy\,:\, y>0\}$ equipped with the metric $d$ derived
from the differential $\rmd s=|\rmd z|/\Imag{z}$ (see the Appendix). Then we apply the
typical methods of variational calculus to the Jacobi functional
\begin{equation*}
J=\int_{P_0}^{P_1}\frac{\sqrt{(\rmd x)^2+(\rmd y)^2}}{y},
\end{equation*}
or, equivalently,
\begin{equation*}
J=\int_{x_0}^{x_1}\frac{\sqrt{1+(y')^2}}{y}\,\rmd x
\end{equation*}
($P_0$ and $P_1$ denote two points of the ambient space where light propagates).
First we prove the following proposition, which refers to the whole upper--half plane $U$.
\begin{proposition}
\label{pro:1}
(i) Let $J$ be the following functional
\beq
\label{quattro}
J=\int_{x_0}^{x_1}\frac{\sqrt{1+(y')^2}}{y}\,\rmd x,
\eeq
and let $F$ denote the integrand of (\ref{quattro}). The Euler--Lagrange equation for
this functional reads
\beq
\label{cinque}
yy'' + y'^2+1=0.
\eeq
(ii) The extremals of functional (\ref{quattro}) are Euclidean half--circles with
centers on the $x$ axis, or Euclidean straight lines normal to the $x$ axis
lying in the half--plane $y>0$. These are the geodesics in the hyperbolic geometry
realized in the half--plane $y>0$. \\
(iii) The Weierstrass condition for the functional (\ref{quattro}) reads
\beq
\label{sei}
F_{y'y'} = \frac{1}{y(1+y'^2)^{3/2}} > 0 ~~~~~ (y>0),
\eeq
and it is satisfied for any $y'$. \\
(iv) There exists a field of extremals of functional (\ref{quattro}), and the
transversality condition becomes an orthogonality condition of these extremals
to the curve $\Phi(x,y)=\Cs$ (constant phase curve), which satisfies the
following equation (eikonal equation):
\beq
\label{otto}
g^{ij}\frac{\partial\Phi}{\partial x_i} \frac{\partial\Phi}{\partial x_j} = 1,
\eeq
where $x_1=x$, $x_2=y$, and $g_{ij}$ is the metric tensor.
\end{proposition}
\begin{proof}
The proof makes use of standard procedures and can be found, for instance,
in Ref. \onlinecite{Smirnov}.
\end{proof}

\noindent
{\bf Remark.}
Let us recall once again that the domain of physical interest where the optical paths necessarily lie
(in view of the fact that $n\geqslant 1$)
is the strip $0<y\leqslant 1$; therefore we shall consider only a subclass of the extremals
of functional (\ref{quattro}): i.e., the half--circles with centers on the
$x$ axis and radius bounded by $0<R\leqslant 1$.

\subsection{Poisson--kernel and horocyclic waves}
\label{subse:poisson}
Let us now give a more precise formulation of the physical problem.
Suppose that a pointlike source of light is pushed to $-\infty$ on the $(x,y)$ plane.
For the sake of simplicity, here we limit ourselves to the scalar representation of light,
and phenomena associated with polarization will not be considered.
From Proposition \ref{pro:1} it follows that light rays are half--circles
with centers on the $x$ axis.
For several reasons which will appear clear
in what follows, it is convenient to map conformally the half--plane $y>0$
into the unit disk $|\zeta|<1$, which amounts to pass from the Poincar\'e
half--plane model $U$ to the Poincar\'e disk model
$D$ (see the Appendix and Fig. \ref{fig_1}). The appropriate conformal mapping is given by:
$\zeta = \rmi(z-\rmi)/(z+\rmi)$ ($z=x+\rmi y$; $\zeta=\xi+\rmi\eta$).

\begin{figure}[tb]
\begin{center}
\leavevmode
\includegraphics[width=9cm]{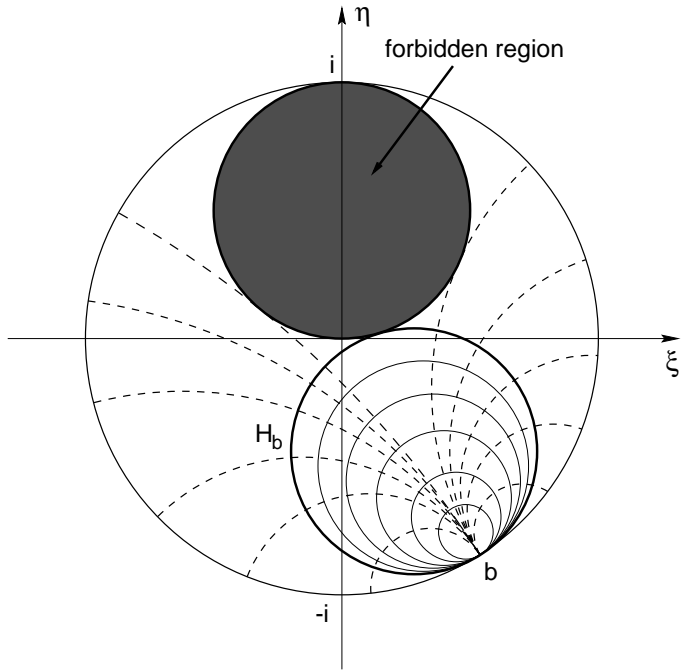}
\caption{\label{fig_1} Horocyclic flow outside the forbidden region in the Poincar\'{e} disk $D$.}
\end{center}
\end{figure}

In the unit disk the light source will be located at $\zeta=\rmi$.
The band $y>1$ will be mapped, in the
$\zeta$--plane, into the disk tangent to the boundary $B$ of $D$ in $\rmi$
with Euclidean radius $\frac{1}{2}$, and represents
the {\it forbidden region} for the light rays. The circular arcs lying in the
half--plane $y>0$  and normal to the $x$ axis will be mapped, in the unit disk,
into circular arcs perpendicular to the boundary $|\zeta|=1$, which are
precisely the geodesics of the hyperbolic geometry in the unit disk model.

From the transversality condition
[see statement (iv) of Proposition \ref{pro:1}], it follows that the
constant phase curve is the curve that intersects orthogonally the
extremals of functional (\ref{quattro}): i.e., the geodesics. In the unit
disk, parallel geodesics are geodesics corresponding to the same point
$b = e^{\rmi\phi}$ on the boundary $B$ of $D$. Therefore, in the physical
problem being treated, the circles tangent to the unit circle at the point
$b$, which intersects orthogonally the pencil of {\it parallel straight lines}
(i.e., arcs of circle orthogonal to $B$) are the constant phase curves,
they are a family of {\it horocycles}, and are denoted by $H_b$.

We can now state the following proposition.
\begin{proposition}
\label{pro:2}
(i) The Poisson kernel
\beq
\label{diciotto}
P(\zeta,b) = \frac{1-|\zeta|^2}{1+|\zeta|^2-2|\zeta|\cos(\theta-\phi)}
~~~~~(\zeta=|\zeta|e^{\rmi\theta};\, b=e^{\rmi\phi}),
\eeq
is constant on each horocycle $H_b$ with normal $b$. \\
(ii) The function
\beq
\label{diciannove}
[P(\zeta,b)]^\nu =
\left[\frac{1-|\zeta|^2}{1+|\zeta|^2-2|\zeta|\cos(\theta-\phi)}\right]^\nu
~~~~~~(\nu \in \C),
\eeq
is an eigenfunction of the Laplace--Beltrami operator on the
hyperbolic disk $D$ corresponding to the eigenvalue $\nu(\nu-1)$. \\
(iii) The hyperbolic waves ({\it horocyclic waves}) are represented
by the following expression:
\beq
\label{venti}
e^{\nu\langle\zeta,b\rangle}=
\left[\frac{1-|\zeta|^2}{1+|\zeta|^2-2|\zeta|\cos(\theta-\phi)}\right]^\nu~~~~~(\nu\in\C),
\eeq
where $\langle\zeta,b\rangle$ is the hyperbolic distance between
the origin of $D$ and the horocycle of normal $b$ passing through
$\zeta \in D$. \\
(iv) The conical functions $\cP_{-\frac{1}{2}+\rmi\lambda}(\cosh r)$
(i.e., the first kind Legendre functions of index
$(-\frac{1}{2}+\rmi\lambda)$ $(\lambda\in\R)$) can be represented by
\beq
\label{ventuno}
\cP_{-\frac{1}{2}+\rmi\lambda}(\cosh r)=\int_B e^{(\frac{1}{2}-\rmi\lambda)\langle\zeta,b\rangle}\,\rmd b
~~~~(\lambda\in\R, B=\{\zeta \,:\, |\zeta|=1\}),
\eeq
and correspond to the fundamental series of the irreducible unitary
representation of the group $SU(1,1)$, which acts transitively on
the hyperbolic disk $D$. \\
(v) The following equality holds:
\beq
\label{26bis}
\cP_{-\frac{1}{2}+\rmi\lambda}(\cosh r)=
\cP_{-\frac{1}{2}-\rmi\lambda}(\cosh r)~~~~(\lambda\in\R).
\eeq
\end{proposition}
\begin{proof}
(i) The level lines of the Poisson kernel $P(\zeta,b)$ are the circles
tangent to the unit circle at the point $b=e^{\rmi\phi}$: i.e., the
images of the horocycles $H_b$ with normal $b$ (see Ref. \onlinecite{Nevanlinna}). \\
(ii) The Laplace--Beltrami operator $\Delta_D$ on the hyperbolic unit
disk $D$ is given by \cite{Helgason1}
\beq
\label{ventidue}
\Delta_D = \frac{1}{4}\left[1-(\xi^2+\eta^2)\right]^2\left(\frac{\partial^2}{\partial \xi^2}
+\frac{\partial^2}{\partial \eta^2}\right).
\eeq
If $\nu \in \C$ is any complex number, a direct computation gives \cite{Helgason1}
\beq
\label{ventitre}
\Delta_D P^\nu(\zeta,b)=\nu(\nu-1) P^\nu(\zeta,b).
\eeq
(iii) In the Euclidean case the function $\bx\to e^{\rmi k(\bx,\bos)}$,
where $k\in\R$, $\bo\in\SSS^{(n-1)}$, $\bx\in\R^n$, represents a plane wave with
normal $\bo$. It is indeed constant on each hyperplane perpendicular to $\bo$,
and furthermore is an eigenfunction of the Laplacian on $\R^n$. The
geometric analog of the plane wave in the case of the hyperbolic disk $D$ is
the function represented by the equality (\ref{venti}) (see Ref. \onlinecite{Helgason2}).
In fact, it is an eigenfunction of the Laplace--Beltrami operator on $D$, as
proved by statement (ii) [see Eq. (\ref{ventitre})]. Further, putting
$\theta=\phi$ in formulae (\ref{diciotto}) and (\ref{diciannove}), we have:
\beq
\label{ventiquattro}
\ln\frac{1-|\zeta|^2}{1+|\zeta|^2-2|\zeta|}=\ln\frac{1+|\zeta|}{1-|\zeta|}=
d(0,\zeta)=\langle |\zeta|e^{i\phi},e^{i\phi} \rangle= \langle \zeta,b \rangle,
\eeq
where $d(0,\zeta)=\ln[(1+|\zeta|)/(1-|\zeta|)]$ is the hyperbolic distance
between the origin and the point $\zeta\in D$ (see the Appendix). Therefore,
$\langle \zeta,b \rangle$ is the hyperbolic analog of $(\bx,\bo)$. In fact,
in view of statement (i), $\langle \zeta,b \rangle$ is the distance between
the origin and the horocycle of normal $b$ passing through $\zeta\in D$,
assuming that the origin falls outside the horocycle; $\langle \zeta,b \rangle$
is positive if the origin is external to the horocycle, while it is negative
($\langle \zeta,b \rangle=\ln[(1-|\zeta|)/(1+|\zeta|)]$) if the origin is
internal to the horocycle. \\
(iv) If we put: $\xi=\tanh(r/2)\cos\theta$, $\eta=\tanh(r/2)\sin\theta$,
then $|\zeta|=\tanh(r/2)$. The Riemannian metric
$\rmd s^2=[4(\rmd\xi^2+\rmd\eta^2)/(1-\xi^2-\eta^2)^2]$ becomes
$\rmd s^2=\rmd r^2+\sinh^2 r\, \rmd\theta^2$. By the use of this substitution in the
expression of the Poisson kernel (\ref{diciotto}) or (\ref{diciannove}), we have:
\beq
\label{venticinque}
\left[\frac{1-|\zeta|^2}{1+|\zeta|^2-2|\zeta|\cos(\theta-\phi)}\right]^\nu=
\frac{1}{[\cosh r-\sinh r \cos(\theta-\phi)]^\nu}~~~~~(\nu\in\C),
\eeq
and the integral sum of {\it horocyclic waves} [see statement (iii)] gives
(see Ref. \onlinecite{Helgason2} and Proposition \ref{pro:6} in the Appendix):
\beq
\label{ventisei}
\int_B e^{\nu \langle \zeta,b \rangle}\,\rmd b=
\frac{1}{2\pi}\int_0^{2\pi}\left(\frac{1}{\cosh r+\sinh r \cos\phi}\right)^\nu\,\rmd\phi=
\cP_{-\nu}(\cosh r)~~~~(\nu\in\C),
\eeq
where $B$ is the boundary of the hyperbolic disk $D$, and
$\cP_{-\nu}(\cosh r)$ are the first kind Legendre functions
\cite{Bateman}. Finally, setting
$\nu=\frac{1}{2}-\rmi\lambda$ $(\lambda\in\R)$ we obtain the
conical functions $\cP_{-\frac{1}{2}+\rmi\lambda}(\cosh r)$,
which correspond to the fundamental series of the irreducible
unitary representation of the group $SU(1,1)$:
i.e., the group of the matrices of the form \cite{Vilenkin}
{\arraycolsep=2pt
\renewcommand{\arraystretch}{0.7}
$\left(\begin{array}{cc} a & c \\ \bar{c} & \bar{a} \end{array}\right)$
},
$|a|^2-|c|^2=1$; $a,c \in \C$, which acts as a group of isometries
of the hyperbolic disk $D$ by means of the map
\beq
\label{ventisette}
g(\zeta) = \frac{a\zeta+c}{\bar{c}\zeta+\bar{a}}~~~~~(\zeta\in D).
\eeq
(v) Equality (\ref{26bis}) is proved in the Appendix (see Proposition \ref{pro:6}).
\end{proof}

~

\noindent
{\bf Remark.} It is well known that the classical Fourier transform refers
to the decomposition of a function, belonging to an appropriate space, into
exponentials of the form $e^{\rmi kx}$ ($k$ real), which can also be viewed
as the irreducible unitary representation of the additive group of real
numbers. Analogously, the exponentials $e^{\rmi (\bk, \bx)}$ are characters
of the group $\R^2$. But the hyperbolic disk is not a group. Therefore a
straightforward generalization of the exponential for $D$ is not possible.
Nevertheless, in view of the fact that the function  $\cP_{-\nu}(\cosh r)$
corresponds to the fundamental series  of the irreducible unitary representation
of the group $SU(1,1)$ for $\nu=\frac{1}{2}-\rmi \lambda$, the exponential
$e^{\left(\frac{1}{2}-\rmi \lambda\right) \langle\zeta, b \rangle}$ ($\lambda\in\R$) represents
the analog of the Euclidean exponential, and plays the same role in the hyperbolic
Fourier analysis \cite{Helgason2}.

\subsection{Conservation of the flow}
\label{subse:conservation}
As already said in the Introduction, the ray trajectories are the lines orthogonal
to the constant phase surface, and are described by the eikonal equation;
moreover, $\langle\zeta, b \rangle$ is the hyperbolic distance between the origin and the
horocycle $H_b$ of normal $b$ passing through $\zeta$. Therefore, in close analogy with
the Euclidean optical geometry, and recalling that
$\cP_{-\frac{1}{2}+\rmi\lambda}(\cosh r)=
\cP_{-\frac{1}{2}-\rmi\lambda}(\cosh r)$ ($\lambda\in\R$) [see statement (v) of Proposition \ref{pro:2}],
the expression of the analog of the Euclidean plane wave $e^{\rmi kx}$ ($k\in\R$) can be written
as follows: $e^{(\frac{1}{2}-\rmi\lambda)\langle\zeta, b \rangle}$
($\lambda\in\R$). Thus the geometrical approximation of the wave function
$\psi$ can be obtained by multiplying $e^{(\frac{1}{2}-\rmi\lambda)\langle\zeta, b \rangle}$
times a function which represents the amplitude. Then we can state the following proposition.
\begin{proposition}
\label{pro:3}
The geometrical approximation of the wave function $\psi$ reads:
\beq
\label{amplitude}
\psi(\zeta, \lambda,b)=A(\lambda)e^{(\frac{1}{2}-\rmi\lambda)\langle \zeta,b \rangle}
~~~~~~(\lambda \in \R, \zeta \in D, b \in B),
\eeq
and the flow in the entire strip $0<y\leqslant 1$ is conserved.
\end{proposition}
\begin{proof}
Let $\sigma$ be the conformal map
\beq
\label{frazione}
z=\sigma(\zeta)=-\rmi\frac{\zeta+\rmi}{\zeta-\rmi},
\eeq
defined in the Appendix, that transfers the geometry of $D$ into $U$.
Since $\sigma(0)=\rmi$ and $\sigma(\rmi)=\infty$, then the image by $\sigma$
of the horocycle $H_\rmi$ passing through $\zeta=0$ is the horizontal line
$\widetilde{H}_\infty=\{x+\rmi y \,:\, y=1\}$ in $U$ (the horocycles in the
Poincar\'e half--plane will be hereafter denoted by $\widetilde{H}_b$).
The image by $\sigma$ of the horocycle $H_{\sigma^{-1}(b)}$ tangent to $H_\rmi$
in $D$ is the horocycle $\widetilde{H}_b$ of radius $\frac{1}{2}$ through
$b \in \R$ and tangent to the horizontal line $\widetilde{H}_\infty$ (in order
to avoid proliferation of notations, we denote by the same letter $b$ both
the points on the boundary $B$ of $D$ and the corresponding points belonging
to the boundary of $\HH^2$, i.e. belonging to $\R$).

We already saw that the horocycle $\widetilde{H}_b$ of normal $b$ is perpendicular
to each geodesic starting from $b$. To calculate the amplitude of the wave function,
we must see how many geodesics perpendicular to $\widetilde{H}_b$ intersect
$\widetilde{H}_b$, with the additional condition that these geodesics belong to the
band $0 < y \leqslant 1$. This corresponds to find the amount of normal vectors at
$\widetilde{H}_b$, with unit norm, that are tangent vectors of geodesics in the
band $0 < y \leqslant 1$.

In general, if $b$ is a point in $\R \cup \{\infty\}$ and $T_1 U$ is the unit
tangent bundle of $U$, then the horocycle flow $h_{j,b}:T_1 U \longrightarrow T_1 U$
is the flow which slides the inward normal vectors to each $\widetilde{H}_b$
to the right along $\widetilde{H}_b$ at unit speed. To find the equation of the
flow $h_{j,b}$, first we consider the flow $h_{j,\infty}$ of geodesics
perpendicular to the horocycle $\widetilde{H}_\infty$ of normal $\infty$.
Then we choose a transformation $M_b$ which maps the horocycle
$\widetilde{H}_\infty$ into the horocycle $\widetilde{H}_b$. In particular,
the map $M_b$ transfers the flow $h_{j,\infty}$ into the flow $h_{j,b}$.\\
From the definition,
\beq
\label{prop31}
h_{j,\infty}(v_i)=
\begin{pmatrix}
1 & j \\
0 & 1
\end{pmatrix}
v_i,
\eeq
where $v_i$ denotes the unit vector vertically upwards based at $i \in U$.
This is because in the simplest case of horocycle flow $h_{j,\infty}$,
the geodesics perpendicular to $\widetilde{H}_\infty$ are vertical lines
and the isometry sending one vertical line into another vertical line is
the horizontal translation. Therefore, the horocycle flow along
$\widetilde{H}_\infty$ is simply the horizontal translation.

Let us now consider the transformation $M_b$ such that $M_b(\infty)=b$.
Then, the horocycle flow $h_{j,b}$ along $\widetilde{H}_b$ is the image
of $h_{j,\infty}$ by $M_b$, hence
\beq
\label{prop32}
h_{j,b}(v)= M_b
\begin{pmatrix}
1 & j \\
0 & 1
\end{pmatrix}
v.
\eeq
It is clear from the definition that the amount of geodesics in the flow
$h_{j,b}$ does not depend on the radius of the horocycle. Given two different
points $b_1$ and $b_2$ in the boundary of the hyperbolic plane, then the
composition of $M_{b_1}$ and $M_{b_2}^{-1}$ sends the point $b_2$ in $b_1$.
Moreover,  $M_{b_1}\circ M_{b_2}^{-1}$ sends the horocycle flow $h_{j,b_2}$
into the horocycle flow $h_{j,b_1}$. This proves that the amplitude of the
wave does not depend on $b$ and $\zeta$.

Using Proposition \ref{pro:2}, we obtain that there exists a function $A(\lambda)$
independent of $\zeta$ and $b$ such that Eq. (\ref{amplitude}) is satisfied, and
the conservation of the flow along the entire strip $0<y\leqslant 1$ is proved.
\end{proof}

~

\noindent
{\bf Remark.} It is interesting to compare the propagation of light in vacuum
with that within the strip $0<y\leqslant 1$ belonging to $\HH^2$. In vacuum each ray cuts
orthogonally all the constant phase planes: i.e., each ray emerging from
a plane cuts orthogonally all the other parallel planes. In $\HH^2$
propagation proceeds in a completely different form. Take two horocycles
lying in the strip $0<y\leqslant 1$, and tangent at the point
$z=(1+\rmi)/2$: the first horocycle, denoted by $\widetilde{H}_0$, has normal $b_0=0$;
the second one, denoted by $\widetilde{H}_1$, has normal $b_1=1$. Only one geodesic,
denoted $\gamma_t$, lying in $\widetilde{H}_0$, cuts orthogonally $\widetilde{H}_1$;
it emerges from $b_0=0$ and ends at $b_1=1$.
All the geodesics $\gamma_>$, emerging from $b_0=0$ and lying in $\widetilde{H}_0$
above $\gamma_t$, cut orthogonally horocycles $\widetilde{H}_b$ with $b>1$;
the geodesics $\gamma_<$, emerging from $b_0=0$ and lying in $\widetilde{H}_0$
below $\gamma_t$, cut orthogonally horocycles $\widetilde{H}_b$ with $b<1$.
However, the density of the flow of geodesics entering orthogonally each horocycle
equals the density of the flow of geodesics exiting orthogonally
the same horocycle.

\section{Transport equation and distribution of the density of trajectories}
\label{se:focusing}
\subsection{Transport equation in the Beltrami pseudosphere}
\label{subse:local}
Working out the problem in the space $\HH^2$ allows us to describe each
trajectory as a geodesic in the Poincar\'{e} plane (or disk), but this setting is
not appropriate for describing the evolution of a bunch of trajectories.
Hereafter we will switch to a representation more suitable for an effective
characterization of the amplitude factor in the geometrical approximation of the field.
To this aim, let us first recall the following well--known negative result due to Hilbert:
there is no regular smooth immersion $X:\HH^2 \rightarrow \R^3$. However,
one can look for a local immersion $X: \cU \rightarrow \R^3$,
where $X$ is a continuous differentiable function, and $\cU\subset\HH^2$
is an open subset. We keep for $\cU$ an open horocycle based at $b$.
This local immersion can be
realized by means of the Beltrami pseudosphere, denoted hereafter by $P_b$
(see the Appendix and Fig. \ref{fig_2}). In fact, let us consider in the hyperbolic
disk $D$ an infinite strip lying between two parallel straight lines
emerging from the source point located on the absolute at $\zeta=-\rmi$.
Then we take on these parallel geodesics a pair of points $A_0$ and $B_0$,
lying on a horocycle of normal $b_0=e^{-\rmi\pi/2}=-\rmi$ and
cutting orthogonally these straight lines; $A_0$ and $B_0$ are spaced at
distance of $2\pi$. One is then led to consider the domain $(-\rmi,A_0,B_0)$.
The Beltrami surface cut along any of its generators can be isometrically
mapped into the domain $(-\rmi,A_0,B_0)$ (see Ref. \onlinecite{Mishchenko}).
On a Lobachevskian plane there always exists reflection (i.e., a hyperbolic
isometry) about an arbitrary straight line; in particular, reflecting the
strip $(-\rmi,A_0,B_0)$ about the straight line $(-\rmi,A_0)$ we obtain a
new strip isometric to the initial one and realized as a cut of the Beltrami
surface in $\R^3$. Reflecting then this new strip $(-\rmi,A_1,A_0)$ (the
segment $A_1A_0$ has length $2\pi$) about the straight line $(-\rmi,A_1)$
we obtain the strip $(-\rmi,A_2,A_1)$ with the same properties.
Exactly the same procedure can be repeated on the other side of
$(-\rmi,A_0,B_0)$, leading to $(-\rmi,B_2,B_1)$. We thus obtain strips of
the form $(-\rmi,A_k,A_{k-1})$ and $(-\rmi,B_k,B_{k-1})$
$(1\leqslant k < \infty)$; all segments $(A_k,A_{k-1})$ and $(B_k,B_{k-1})$
have the same length $2\pi$. Working with the same procedure we can now
construct the map of the open
horocycle $H_{b_0}$, tangent at the boundary to the forbidden region (this latter
represented by the horocycle $H_\rmi$ of normal $\rmi$ and passing
through the origin), into a Beltrami funnel, such that each strip of
the type $(-\rmi,A_k,A_{k-1})$, $(-\rmi,A_0,B_0)$, and $(-\rmi,B_k,B_{k-1})$,
$(1\leqslant k < \infty)$ (referred, now, to the horocycle $H_{b_0}$), is mapped
isometrically into the Beltrami
surface, the horocycle $H_{b_0}$ being wound infinitely many times into
the Beltrami surface \cite{Mishchenko} (see Fig. \ref{fig_2}).
We can repeat the same procedure for each point $b \in B$, since there
is a rotation (i.e., a hyperbolic isometry) sending each $b \in B$ onto $b_0$.\\
For an explicit equation of the immersion $X$, the reader is referred to Ref. \onlinecite{Treibergs}.

\begin{figure}[tb]
\begin{center}
\leavevmode
\includegraphics[width=5cm]{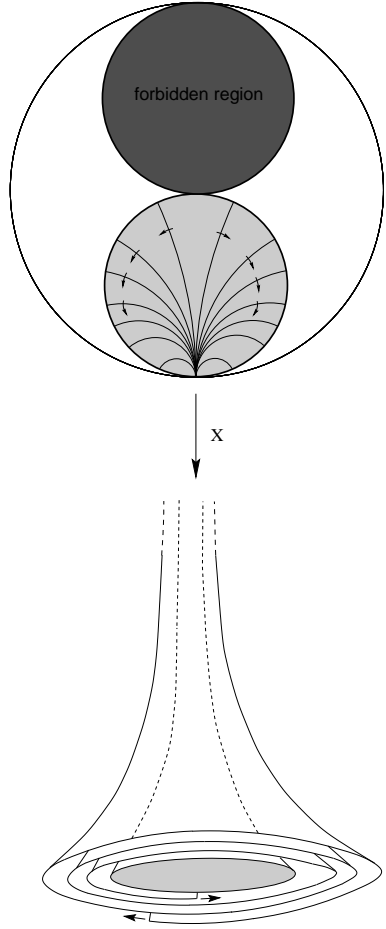}
\caption{\label{fig_2} Mapping of a horocycle in the disk $D$ into a Beltrami pseudosphere.}
\end{center}
\end{figure}

~

In general, the Laplace--Beltrami operator $\Delta_M$ on a two--dimensional
Riemannian manifold $M$ with metric tensor $g_{ij}$
($g=|\mbox{det}(g_{ij})|$, $g^{ij}=g^{-1}_{ij}$) is defined as follows:
\beq
\label{trentadue}
\Delta_M =
\frac{1}{\sqrt{g}}\left[\sum_{i=1}^2\frac{\partial}{\partial x_i}
\left(\sum_{j=1}^2 g^{ij}\sqrt{g}\frac{\partial}{\partial x_j}\right)\right].
\eeq
In the specific case of the hyperbolic metric associated with
the refractive index $n(y)=1/y$ (see the Appendix), the Laplace--Beltrami operator reads:
\beq
\label{laplacian}
\Delta_\HH=\frac{1}{n^2}\left(\frac{\partial^2}{\partial x^2}+
\frac{\partial^2}{\partial y^2}\right)=
y^2\left(\frac{\partial^2}{\partial x^2}+\frac{\partial^2}{\partial y^2}\right).
\eeq
We then have the following proposition.
\begin{proposition}
\label{pro:4}
(i) The Helmholtz equation reads
\beq
\label{ventinove}
\Delta_\HH \psi + k_\HH^2\psi=0,
\eeq
where $k_\HH^2=\lambda^2+\frac{1}{4}$ $(\lambda\in\R)$.\\
(ii) The geometrical approximation of the wave function $\psi$ (for $|\lambda|\to\infty$),
written in terms of the Beltrami coordinates (see the Appendix), reads
\beq
\label{trenta}
\psi_\pm(\lambda,u)=C(\lambda)\,e^{u/2}\,e^{\mp\rmi\lambda u}~~~~~(\lambda\in\R; u\geqslant 0).
\eeq
\end{proposition}
\begin{proof}
(i) Let us consider the {\it horocyclic waves} which generate the
conical functions $\cP_{-\frac{1}{2}\pm\rmi\lambda}(\cosh r)$,
corresponding to the irreducible unitary representation of the
$SU(1,1)$ group, which acts transitively on the hyperbolic disk $D$.
This amounts to put in the exponent $\nu\in\C$ of the Poisson kernel:
$\nu=\frac{1}{2}\pm\rmi\lambda$ $(\lambda\in\R)$. Accordingly,
the {\it horocyclic waves} read
$e^{(\frac{1}{2}\pm\rmi\lambda)\langle\zeta,b\rangle}$ [see statements (iv) and (v)
of Proposition \ref{pro:2}]. From statement (ii) of Proposition \ref{pro:2}
and Eq. (\ref{ventitre}) we get:
\beq
\label{trentuno}
\Delta_\HH \,e^{(\frac{1}{2}\pm\rmi\lambda)\langle\zeta,b\rangle}=
-\left(\lambda^2+\frac{1}{4}\right) e^{(\frac{1}{2}\pm\rmi\lambda)\langle\zeta,b\rangle}
=-k_\HH^2\,e^{(\frac{1}{2}\pm\rmi\lambda)\langle\zeta,b\rangle},
\eeq
where $k_\HH^2=\lambda^2+\frac{1}{4}$ $(\lambda\in\R)$. Next, proceeding
in close analogy with the Euclidean case, where the Euclidean plane wave
plays the role of the {\it horocyclic wave}, we obtain Eq. (\ref{ventinove}). \\
(ii) Let us now go back to the mapping of the horocycle into the Beltrami
funnel (without a cut) in $\R^3$, illustrated above. Next, we apply the
Laplace--Beltrami operator to the wave function $\psi$, supposed to belong
to $C^\infty(\cB)$ ($\cB$ denoting the Beltrami pseudosphere);
in (\ref{trentadue}) $x_i$ ($i=1,2$) stand for the Beltrami coordinates
$u,v$. Recall that the first fundamental form in Beltrami coordinates
reads [see part (C) of the Appendix]:
\beq
\label{beltr}
I=\rmd u^2+e^{-2u}\rmd v^2~~~~~(u\geqslant 0).
\eeq
Accordingly, we have $g_{11}=1$, $g_{22}=e^{-2u}$, $g_{12}=g_{21}=0$,
$g=|\det(g_{ij})|=e^{-2u}$, $g^{ij}=g_{ij}^{-1}$.
Thus, we are led to the following equation:
\beq
\label{trentatre}
\Delta_\cB\psi+k_\HH^2\psi=0,
\eeq
where $\Delta_\cB$ is the Laplace--Beltrami operator, referred to the
Beltrami pseudosphere. In this equation, we pass from the coordinates
$(\xi, \eta)$ of the hyperbolic disk $D$ to the Beltrami coordinates $(u,v)$
of the Beltrami pseudosphere. We illustrate with more details this passage.
First we embed an open horocycle $H_b$ of normal $b$ and tangent to the
{\it forbidden region} (represented by the horocycle $H_{\rmi}$ passing
through the origin of $D$ and with normal $\rmi$; see Fig. \ref{fig_1}) into a
Beltrami pseudosphere.
Notice that in the present analysis, as well as in Proposition \ref{pro:3},
and in strict analogy with the classical Euclidean procedure,
we consider the distance from the origin of the hyperbolic disk $D$
(rather than from the point source located at $\zeta=\rmi$) to
the horocycle $H_{\zeta,b}$ (inside $H_b$) of normal $b$ passing through a point $\zeta$.
Thus we have
\beq
\langle\zeta,b\rangle=d(0,H_b)+d(H_b,H_{\zeta,b}):=d_b+d(H_b,H_{\zeta,b}).
\eeq
When we embed $H_b$ into a Beltrami pseudosphere, the distance
$d(H_b,H_{\zeta,b})$ between horocycles corresponds to the distance
between different parallels $u=\Cs$ inside the pseudosphere.
Since $H_b$ is fixed, then $d_b$ is fixed too. Then, following
the standard method of stationary phase, we look now for a solution
of equation (\ref{trentatre}), of the following form:
\begin{eqnarray}
\label{trentaquattro}
\psi(\lambda,\bx) &=& \int A(\bx,\ell)\,e^{(\frac{1}{2}-\rmi\lambda)\langle\zeta,b\rangle}\,\rmd\ell
=e^{(\frac{1}{2}-\rmi\lambda)d_b}\int A(\bx,\ell)\,e^{(\frac{1}{2}-\rmi\lambda)d(H_b,H_{\zeta,b})}\,\rmd\ell \nonumber \\
&=& C(\lambda)\int A(\bx,\ell)\,e^{(\frac{1}{2}-\rmi\lambda)\Phi(\bx,\ell)}\,\rmd\ell,
\end{eqnarray}
where $\bx=(x_1,x_2)$, $x_1=u$, $x_2=v$; $\ell$ is the pathlength inside the
pseudosphere to the point of coordinate $\bx$, $\Phi(\bx,\ell)$ denotes the
phase (recall the statements of Proposition \ref{pro:2}).

The r.h.s. of equality (\ref{trentaquattro}) is an integral of oscillating type.
The principal contribution to $\psi(\lambda,\bx)$, as $|\lambda|\rightarrow +\infty$,
corresponds to the stationary point of $\Phi$, in the neighborhood of which
the exponential ceases to oscillate rapidly. These stationary points can be
obtained from the equation $\partial\Phi/\partial \ell=0$ (provided that
$\partial^2\Phi/\partial \ell^2 \neq 0$). If the condition
$\partial\Phi/\partial \ell=0$ is satisfied by a unique value
$\ell_0$ of $\ell$, corresponding to the unique ray trajectory (geodesic) passing
across the point of coordinates $(u,v)$, we say that $\Phi$ has a critical
nondegenerate point at $\ell=\ell_0$. Moreover, recalling that the manifolds
with nonpositive curvature do not have conjugate points, we can state that
all the critical points of $\Phi$ are nondegenerate. Then, by applying the
Morse lemma on the representation of the functions all of whose critical
points are nondegenerate, we obtain the following
asymptotic evaluation of integral (\ref{trentaquattro}):
\beq
\label{trentacinque}
\psi(\lambda,\bx) = C(\lambda)e^{\frac{1}{2}\Phi(\bx,\ell_0)}\,e^{-\rmi\lambda\Phi(\bx,\ell_0)}
\sum_{m=0}^\infty \frac{A_m(\bx)}{(\rmi\lambda)^m}.
\eeq
The leading term of expansion (\ref{trentacinque}) reads
\beq
\label{trentasei}
\psi(\lambda,\bx) = C(\lambda)A_0(\bx) e^{\frac{1}{2}\Phi(\bx,\ell_0)} e^{-\rmi\lambda\Phi(\bx,\ell_0)},
\eeq
where
\beq
\label{trentasette}
A_0(\bx)=A(\bx,\ell_0)\left(\left|\frac{\partial^2\Phi}{\partial \ell^2}\right|^{-1/2}\right)_{\ell=\ell_0}
\exp\left[\rmi\,\frac{\pi}{4}\,\Sg\left(\frac{\partial^2\Phi}{\partial \ell^2}\right)\right]_{\ell=\ell_0}.
\eeq
For simplicity, in the following we shall write the leading term of expansion
(\ref{trentacinque}) as: $C(\lambda)A(\bx)e^{(\frac{1}{2}-\rmi\lambda)\Phi(\bx)}$,
dropping the zero subscripts. Substituting this expression into Eq. (\ref{trentatre}),
collecting powers of $(\rmi\lambda)$ and, finally, equating to zero their coefficients,
two equations are obtained: the eikonal (or Hamilton--Jacobi) equation
\beq
\label{trentotto}
g^{ij}\frac{\partial\Phi}{\partial x_i} \frac{\partial\Phi}{\partial x_j} = 1,
\eeq
and the transport equation
\beq
\label{trentanove}
\frac{1}{\sqrt{g}}\sum_{i=1}^2\frac{\partial}{\partial x_i}\left[\sqrt{g}A^2\,e^\Phi\sum_{j=1}^2 g^{ij}
\frac{\partial\Phi}{\partial x_j}\right]=0.
\eeq
Let us note that in the present problem the wave functions are radial,
in view of the fact that we are considering a family of horocycles having
all the same normal $b$. Therefore, $\psi(\lambda,\bx)$ [where $\bx\equiv(u,v)$,
$u,v$ being the Beltrami coordinates, see the Appendix] does not depend on $v$.
Then, Eq. (\ref{trentotto}) becomes:
\beq
\label{quaranta}
\left(\frac{\rmd\Phi}{\rmd u}\right)^2 = 1,
\eeq
which gives the following expression of the
phase: $\Phi^{(\pm)}=\pm u + c$ ($c=\Cs$). Proceeding analogously with Eq. (\ref{trentanove}),
we have:
\beq
\label{quarantuno}
\frac{\rmd}{\rmd u}\left(A^2\,e^{\pm u}\,e^{-u}\frac{d\Phi^{(\pm)}}{du}\right) = 0.
\eeq
Substituting in the leading term (\ref{trentasei}) the expressions of
$\Phi$ and $A$, which derive from (\ref{quaranta}) and (\ref{quarantuno}),
the rhs of (\ref{trenta}) follows.
\end{proof}

\subsection{Sine--Gordon equation and the flow of trajectories}
\label{subse:sine}
The analysis of propagation in the Beltrami pseudosphere has allowed us to study the
distribution of the density of trajectories and, accordingly, the ray focusing along the horizontal axis.
Another parameter, related once again to the Beltrami pseudosphere, and whose characterization
is relevant in our description of the flow of trajectories,
is the angle $\varphi$ that the tangent to the meridian (in the Beltrami
pseudosphere) makes with the $z$ axis [see section (C) of the Appendix].

In the Appendix, the Beltrami pseudosphere is described in terms of
the Beltrami coordinates $u$ and $v$. Here we choose another parameterization $(p,q)$ by setting
\begin{subequations}
\begin{eqnarray}
\label{cambio_insieme}
\rmd p&=&-\csc{\varphi}\,\rmd\varphi, \label{cambio} \\
\rmd q&=&-\rmd v \label{cambio02}.
\end{eqnarray}
\end{subequations}
Integrating (\ref{cambio}), we obtain:
\beq
\label{solution2}
\varphi=2\tan^{-1}{\left(e^{-p}\right)},
\eeq
and
\beq
\label{calcolo}
\frac{\rmd\varphi}{\rmd p}=-\sin{\varphi}=\frac{1}{\cosh p}.
\eeq
Next, substituting in formula (\ref{rotazione2}) of the Appendix (with $\rho=-1$)
$\varphi$ and $v$ in terms of the parameters $p$ and $q$, the position vector
$\br$ of the pseudosphere can be rewritten as follows:
\beq
\label{rotazionepq}
\br(p,q)=\left(
\begin{array}{ccc}
-\frac{1}{\cosh p}\cos q \\
\frac{1}{\cosh p}\sin q \\
p-\tanh{p}
\end{array}
\right),
\eeq
where the downward vertex of the pseudosphere corresponds to $p\rightarrow -\infty$,
while the rim corresponds to $p=0$. The parameters $(p,q)$ can be related to
the arc lengths $(\alpha, \beta)$ along asymptotic lines  as follows:
\beq
\label{cambio2}
\begin{array}{ll}
p=\alpha+\beta, \\
q=\alpha-\beta.
\end{array}
\eeq
Setting $\omega=2\varphi$, the first fundamental form reads
\beq
\label{formarotazione6}
I=\rmd\alpha^2+2\cos{\omega}\rmd\alpha\,\rmd\beta+\rmd\beta^2,
\eeq
which can be easily derived from the expression (\ref{formarotazione3}) of
the Appendix through the formulae (\ref{cambio}) and (\ref{cambio2})
(recalling that $\rho^2=1$). Moreover, from Gauss' and Weingarten's equations,
it follows that $\omega$, which is the angle between the asymptotic lines,
satisfies the {\it classical} sine--Gordon equation, which reads
\beq
\label{sinegordon3}
\frac{\partial^2\omega}{\partial\alpha\,\partial\beta}=\sin{\omega}.
\eeq
This latter equation rewritten in terms of the coordinates $p,q$ becomes
\beq
\label{sinegordon2}
\frac{\partial^2 \varphi}{\partial p^2}-\frac{\partial^2 \varphi}{\partial q^2}=\sin{\varphi}\cos{\varphi}.
\eeq
We can finally state the following proposition:
\begin{proposition}
\label{pro:5}
(i) The angle $\varphi$, that the tangent to the meridian of the pseudosphere
makes with the $z$ axis (i.e., the rotation axis), is represented by the
following formula:
\beq
\label{solution}
\varphi=2\tan^{-1}{\left(e^{-p}\right)},
\eeq
which is the so--called ``one--soliton" solution of the sine--Gordon
equation (\ref{sinegordon2}). \\
(ii) The angle $\varphi$ varies from $\varphi=\frac{\pi}{2}$ to $\varphi=\pi$
in the process of focusing, and from $\varphi=\pi$ to
$\varphi=\frac{\pi}{2}$ in the process of defocusing.
\end{proposition}
\begin{proof}
(i) The proof of this statement follows easily by direct calculation. \\
(ii) The downward vertical $z$ axis of the pseudosphere we are considering
is negatively oriented; then varying $p$ from $0$ to $-\infty$
[in equation (\ref{solution})], $\varphi$ varies from $\varphi=\frac{\pi}{2}$
to $\varphi=\pi$ (focusing). Next varying $p$ from $-\infty$ to $0$,
$\varphi$ varies from $\varphi=\pi$ to $\varphi=\frac{\pi}{2}$ (defocusing).
\end{proof}

So far in this section we have considered a local description of the flow
in order to show that the density of the flow is not homogeneously distributed.
We now want to recover a global description in the entire strip $0<y\leqslant 1$.
But we have already remarked that the map of a horocycle into a pseudosphere
excludes the boundary of the horocycle. Therefore this
global description cannot be reached by considering only horocycles lying in the strip $0<y\leqslant 1$
and tangent to the line $y=1$. Recall, however, that while the \emph{physical}
geodesics, which lie within the strip $0<y\leqslant 1$ in $\HH^2$,
cannot enter the forbidden region $y>1$, this is not the case for horocycles,
provided we limit ourselves to consider in these horocycles
(i.e., those entering the forbidden region) those segments of geodesics which lie in the strip $0<y\leqslant 1$.
In view of these considerations it is sufficient to consider horocycles
$\widetilde{H}'_b \supset \widetilde{H}_b$ ($\widetilde{H}'_b$ entering the forbidden region),
and, accordingly, the maps $X'_b$ embedding $\widetilde{H}'_b$ into a Beltrami
pseudosphere $P'_b$. We can thus obtain a sequence of Beltrami pseudospheres
$P'_b$ which, by varying $b$, allows us to connect in $\R^3$ the geodesics lying in the
strip $0<y\leqslant 1$. We thus pass from a local to a global description.

\section{Conclusions and discussion}
\label{se:conclusions}
In this paper we have presented the geometrical optics generated by a refractive
index of hyperbolic type.
The ray trajectories are geodesics in the Poincar\'e--Lobachevsky half--plane,
and the horocyclic waves, which are related to the Poisson kernel, represent the
analogs of the Euclidean plane waves.
We thus obtain two main results:
\begin{itemize}
\item[(a)] The flow in the entire strip $0<y\leqslant 1$ is conserved (see Proposition \ref{pro:3});
\item[(b)] inside each horocycle (by embedding horocycles in Beltrami pseudospheres)
the ray focusing on each point $b$ of the horizontal $x$ axis: i.e., toward
the boundary of $\HH^2$, is shown.
\end{itemize}

The connection between hyperbolic geometry and optics of spatially
nonuniform media can be made even tighter.
In fact, it can be shown that the transfer matrix associated
with lossless layered optical media is an element of the group $SU(1,1)$,
no matter how complicated the stepwise profile of the refractive index
might be \cite{Barriuso,Monzon,Yonte}. Therefore, the action of any
lossless optical multilayer can be regarded as a M\"{o}bius transformation
on the unit disk \cite{Barriuso}, and therefore the natural geometric environment
for these physical systems is the hyperbolic one.
From this point of view, the geometrical optics description of light propagation
in the ``hyperbolic glass'' discussed so far can be regarded as
the study (in a spatially continuum setting) of the particular case of special interest,
in which only motions along hyperbolic geodesics are allowed.

\section*{Appendix}
\label{se:appendix}
\setcounter{equation}{0}
\renewcommand{\theequation}{A.\arabic{equation}}
\noindent
{\bf (A)} Let us consider the upper half--plane model of the hyperbolic
two--dimensional space $U=\{z=x+\rmi y \,:\, y>0\}$.
Then the boundary $\partial U$ of $U$ is the real axis and infinity.
On $U$ we can define a metric $d$ derived from the differential
$\rmd s=|dz|/\Imag{z}$ where $\rmd z$ is the standard Euclidean metric.
The geodesics of $U$ are vertical half--lines and Euclidean semicircles with
center on the real axis. The group of the orientation preserving isometries of
$(U,d)$ is the M\"{o}bius group $PSL_2(\R)$, that is the group of the $2\times 2$
matrices of real coefficients with determinant $1$. The action on
$U\cup\partial U$ is defined as $\gamma(z)=(az+b)/(cz+d)$, where
{
\arraycolsep=2pt
\renewcommand{\arraystretch}{0.7}
$\left(
\begin{array}{cc}
a & b \\
c & d
\end{array}
\right)$
}
$\in PSL_2(\R)$.

Another model of the hyperbolic two--dimensional space is the Poincar\'{e}
disk $D=\{\zeta=\xi+\rmi\eta \,:\, \xi^2+\eta^2 <1\}$. The map
\beq
\label{map_U_D}
\zeta=\rmi \frac{z-\rmi}{z+\rmi},
\eeq
\{with inverse $z= -\rmi[(\zeta+\rmi)/(\zeta-\rmi)]$\} transfers the geometry
of $U$ into the geometry of $D$. In particular, the metric on $D$ is given by
the differential
\beq
\label{diff1}
\rmd\zeta=\frac{2}{1-|\zeta|^2}\rmd_E\zeta,
\eeq
where $\rmd_E$ is the standard Euclidean metric. In this model, the geodesics are
circular arcs perpendicular to the boundary $|\zeta|=1$. Moreover, the group
of the orientation preserving isometries of the Poincar\'e disk $D$ is the
group $SU(1,1)$ of the maps of the following form:
\beq
\label{mappi}
\left\{\frac{a\zeta+\overline{c}}{c\zeta+\overline{a}} \,:\, |a|^2-|c|^2=1 \right\}.
\eeq

~

\noindent
{\bf (B)}
The group $G=SU(1,1)$ admits two subgroups relevant for our analysis:
\begin{enumerate}
\item[(1)]
The subgroup $K$ of rotations
\begin{equation*}
k_{\theta}=
\left(
\begin{array}{cc}
e^{\rmi\frac{\theta}{2}} & 0 \\
0 & e^{-\rmi \frac{\theta}{2}}
\end{array}
\right)
\qquad
(0 \leqslant \theta < 4\pi).
\end{equation*}
\item[(2)] The subgroup $A$ of matrices
\begin{equation*}
a_r=
\left(
\begin{array}{cc}
\cosh{\frac{r}{2}} & \sinh{\frac{r}{2}} \\
\sinh{\frac{r}{2}} & \cosh{\frac{r}{2}}
\end{array}
\right)
\qquad (r \in \R).
\end{equation*}
\end{enumerate}
Let $A^+$ denote the set $a_r$ with $r\geqslant 0$. Then the following decomposition holds.

\noindent
{\bf Cartan decomposition \cite{Helgason1}}:
Any element $g\in G$ can be decomposed as follows:
\beq
\label{novantotto}
g = k_\theta a_r k_\phi~~~~~(r\geqslant 0,\,0\leqslant\theta< 4\pi,\,0\leqslant\phi < 2\pi),
\eeq
i.e., $G = K A^+ K$. The decomposition is unique if $g \not\in K$.

We now introduce the so--called {\it spherical functions} $\Phi_\nu(g)$ on
$G/K$ [$g\in G=SU(1,1),\,\nu\in\C$], which are defined as follows \cite{Eymard}.
\begin{definition}
\label{def:1}
The spherical functions on $G/K$ $[G=SU(1,1),\,K=SO(2)]$ are defined as follows:
\beq
\label{cento}
\Phi_\nu(g)=\int_B \left|\frac{\rmd(g^{-1}\cdot b)}{\rmd b}\right|^\nu\,\rmd b~~~~~[g\in SU(1,1),\,\nu\in\C],
\eeq
where $B$ is the boundary of the hyperbolic disk $D$ (i.e., $B=\{\zeta \,:\, |\zeta| = 1\}$).
\end{definition}
We can then prove the following proposition.
\begin{proposition}
\label{pro:6}
The functions $\Phi_\nu(g)$ satisfy the following properties:\\
(i)
\beq
\label{centouno}
\Phi_\nu(g)= \cP_{-\nu}(\cosh r),~~~~~[g\in SU(1,1),\,\nu\in\C],
\eeq
where $\cP_{-\nu}(\cdot)$ are the first kind Legendre functions. \\
(ii)
\beq
\label{centodue}
\Delta_D \Phi_\nu(g) = \nu(\nu-1)\Phi_\nu(g),~~~~~[g\in SU(1,1),\,\nu\in\C],
\eeq
where $\Delta_D$ is the hyperbolic Laplace--Beltrami operator. \\
(iii)
\beq
\label{app8}
\cP_{-\frac{1}{2}+\rmi\lambda}(\cosh r)=
\cP_{-\frac{1}{2}-\rmi\lambda}(\cosh r)~~~~(\lambda\in\R).
\eeq

\end{proposition}
\begin{proof}
Let us consider the following integral $\frac{1}{2\pi}\int_0^{2\pi}f(e^{\rmi\phi})\,\rmd\phi$
[$f\in L^1(B)$], and evaluate how the Lebesgue measure $\frac{1}{2\pi}\rmd\phi$ changes when an element
$g \in SU(1,1)$ acts on $B$ (boundary of $D$). Recalling that the action of $g$ is
$g\cdot \zeta = (a\zeta+c)/(\bar{c}\zeta+\bar{a})$ ($\zeta\in D$) (see Proposition \ref{pro:2}), we have
$g\cdot e^{\rmi\phi}=e^{\rmi\chi}=(ae^{\rmi\phi}+c)/(\bar{c}e^{\rmi\phi}+\bar{a})$.
We thus have \cite{Eymard}:
\beq
\label{centotre}
\left|\frac{\rmd\phi}{\rmd\chi}\right|=\left|\bar{c}e^{\rmi\chi}-a\right|^{-2}=
|a|^{-2}\left|\frac{\bar{c}}{a}e^{\rmi\chi}-1\right|^{-2}=
\left(1-\frac{|c|^2}{|a|^2}\right) \left|\frac{\bar{c}}{a}e^{\rmi\chi}-1\right|^{-2},
\eeq
since $|a|^2-|c|^2=1$. Let us now note that $g\cdot 0 = c/\bar{a}$; on the other hand, in view of
the Cartan decomposition, we have:
\beq
\label{centoquattro}
g\cdot 0 = k_\theta \, a_r \cdot 0 = e^{\rmi\theta}\tanh\left(\frac{r}{2}\right)=|\zeta|e^{\rmi\theta}.
\eeq
Therefore $(1-|c|^2/|a|^2)=1-|\zeta|^2$, and
\beq
\label{centocinque}
\left|\frac{\bar{c}}{a}e^{\rmi\chi}-1\right|^{-2}=
\frac{1}{\left|1-|\zeta|e^{\rmi(\chi-\theta)}\right|^2}=
\frac{1}{1+|\zeta|^2-2|\zeta|\cos(\chi-\theta)}.
\eeq
We thus have
\beq
\label{centosei}
\frac{1}{2\pi}\int_0^{2\pi} f(g\cdot e^{\rmi\phi})\,\rmd\phi=
\frac{1}{2\pi}\int_0^{2\pi} f(e^{\rmi\chi})\left|\frac{\rmd\phi}{\rmd\chi}\right|\,\rmd\chi,
\eeq
and, in view of formulae (\ref{centotre}), (\ref{centoquattro}), (\ref{centocinque}), (\ref{centosei}), we get
\beq
\label{centosette}
P(g\cdot 0,b) = \frac{1-|\zeta|^2}{1+|\zeta|^2-2|\zeta|\cos(\chi-\theta)}=
\left|\frac{\rmd(g^{-1}\cdot b)}{\rmd b}\right|.
\eeq
Recalling Definition \ref{def:1} we can write
\beq
\label{centootto}
\Phi_\nu(g) = \int_B \left|\frac{\rmd(g^{-1}\cdot b)}{\rmd b}\right|^\nu\,\rmd b =
\frac{1}{2\pi}\int_0^{2\pi}\left(\frac{1-|\zeta|^2}{1+|\zeta|^2-2|\zeta|\cos\phi}\right)^\nu\,\rmd\phi.
\eeq
Finally, writing $\tanh(r/2)$ in place of $|\zeta|$ [see formula (\ref{centoquattro})], we obtain
\beq
\label{centonove}
\Phi_\nu(g) = \frac{1}{2\pi}\int_0^{2\pi}\frac{1}{(\cosh r + \sinh r \cos\phi)^\nu}\,\rmd\phi =
\cP_{-\nu}(\cosh r),
\eeq
where the last equality follows from the integral representation of the first kind Legendre functions \cite{Bateman}.
Formula (\ref{centonove}) proves statement (i).

From formula (\ref{centonove}) it follows that the spherical functions $\Phi_\nu(g)$ are bi--$k$--invariant;
indeed, using the Cartan decomposition, we have
\beq
\label{centodieci}
\Phi_\nu(g) = \Phi_\nu(k_\theta\, a_r\,k_\phi) = \Phi_\nu(a_r) = \cP_{-\nu}(\cosh r).
\eeq
Furthermore, we can prove that [see formulae (\ref{centosette}) and (\ref{centootto})]:
\beq
\label{centoundici}
\Phi_\nu(E)=\int_B\left[P(E\cdot 0,b)\right]^\nu\,\rmd b=1=\cP_{-\nu}(1).
\eeq
Let us now return to statement (ii) of Proposition \ref{pro:2}, and recall that $P^\nu(\zeta,b)$
is an eigenfunction of the hyperbolic Laplace--Beltrami operator with eigenvalue $\nu(\nu-1)$. We then
consider an integral of the following form:
\beq
\label{centododici}
\int_B\left[P(g\cdot 0,b)\right]^\nu\,\rmd b = \Phi_\nu(g).
\eeq
It can be proved that this integral superposition is still an eigenfunction of the
hyperbolic Laplace--Beltrami operator $\Delta_D$, having $\nu(\nu-1)$ as eigenvalue \cite{Eymard}. In particular,
it follows that $\Phi_\nu(g)$ and $\Phi_{1-\nu}(g)$ [$g\in SU(1,1)$] are both eigenfunctions of $\Delta_D$
with the same eigenvalue, and therefore they coincide. We can thus state that
$\cP_{-\frac{1}{2}+\rmi\lambda}(\cosh r) = \cP_{-\frac{1}{2}-\rmi\lambda}(\cosh r)$;
statements (ii) and (iii) are thus proved.

Finally, from the representation of the Poisson kernel in terms of {\it horocyclic waves},
we have
\begin{eqnarray}
\label{centotredici}
\lefteqn{\cP_{-\frac{1}{2}+\rmi\lambda}(\cosh r) =
\frac{1}{2\pi}\int_0^{2\pi}\left(\frac{1}{\cosh r + \sinh r \cos\phi}\right)^{1/2-\rmi\lambda}\rmd\phi} \nonumber \\
&&=\frac{1}{2\pi}\int_0^{2\pi}\!\!\left(\frac{1-|\zeta|^2}{1+|\zeta|^2-2|\zeta|\cos\phi}\right)^{1/2-\rmi\lambda}\!\!\!\rmd\phi=
\int_B e^{(\frac{1}{2}-\rmi\lambda)\langle \zeta,b \rangle}\rmd b~~~(\lambda\in\R,\,\zeta\in D).
\end{eqnarray}
\end{proof}

\noindent
{\bf (C)} In Sec. \ref{se:focusing}, we have seen that it is possible to have a local
isometric immersion from an horocycle to $\R^3$. The image of this immersion
is a Beltrami pseudosphere. Now, we describe in more details the equation
of the pseudosphere and its coordinates.

First, we want to recover the pseudosphere as a surface of revolution of
a curve in $\R^3$. Let us recall that, in general, the position vector $\br$
of the surface of revolution generated by the rotation of a plane curve
$z=\Phi(r)$ about the $z$ axis is given by:
\beq
\label{rotazione}
\br=
\left(
\begin{array}{ccc}
r\cos{v} \\
r\sin{v} \\
\Phi(r)
\end{array}
\right),
\eeq
where $v$ varies between $0$ and $2\pi$. Here the circles $r=\Cs$
are the parallels and the curves $v=\Cs$ are the meridians.
The first fundamental form associated with the surface (\ref{rotazione})
is given by
\beq
\label{formarotazione}
I=\left\{1+[\Phi'(r)]^2\right\}\rmd r^2+r^2 \rmd v^2.
\eeq
We now rewrite the form (\ref{formarotazione}) as follows:
\beq
\label{formarotazione2}
I=\rmd u^2+r^2 \rmd v^2,
\eeq
where
\beq
\label{sostituzione}
\rmd u=\sqrt{1+[\Phi'(r)]^2}\rmd r; ~~~~~ r=r(u).
\eeq
From the general Gauss' theory of surfaces, we have that the total
curvature is given by \cite{Rogers}
\beq
\label{curvaturabacklund}
K=-\frac{1}{r}\frac{\rmd^2r}{\rmd u^2},
\eeq
whence the general pseudospherical surface of revolution with
$K=-1/\rho^2$ adopts the form \cite{Rogers}
\beq
\label{posizionepseudo}
r(u)=c_1\cosh{\frac{u}{\rho}}+c_2\sinh{\frac{u}{\rho}}.
\eeq
In the case $c_1=c_2=c$, which corresponds to a parabolic pseudospherical surface of revolution,
the meridians are given by
\beq
\label{meridians}
r(u)=ce^{u/\rho},
\eeq
while
\beq
\label{sostituzione2}
z=\Phi(r)=\int\sqrt{1-\left(\frac{c}{\rho}\right)^2e^{2u/\rho}}\,\rmd u.
\eeq
Then, the first fundamental form, with $c=1$, has the following expression:
\beq
\label{formafinale}
I=\rmd u^2+e^{2u/\rho}\,\rmd v^2.
\eeq
The coordinates $u,v$ are called \emph{Beltrami coordinates}.

The substitution
\beq
\label{sostituzione3}
\sin{\varphi}=\frac{c}{\rho}e^{u/\rho}
\eeq
in (\ref{sostituzione2}) yields
\beq
\label{trattrice2}
z=\rho\left(\cos{\varphi}+\ln\left|\tan{\frac{\varphi}{2}}\right|\right).
\eeq
From formulae (\ref{rotazione}), (\ref{meridians}), (\ref{sostituzione3})
and (\ref{trattrice2}) we obtain
\beq
\label{rotazione2}
\br=\left(
\begin{array}{ccc}
\rho\sin\varphi\cos v \\
\rho\sin\varphi\sin v \\
\rho\left(\cos\varphi+\ln\left|\tan{\frac{\varphi}{2}}\right|\right)
\end{array}
\right),
\eeq
and the first fundamental form (in terms of $\varphi$ and $v$) is
\beq
\label{formarotazione3}
I=\rho^2\cot^2\varphi\,\rmd\varphi^2+\rho^2\sin^2\varphi\,\rmd v^2.
\eeq

Equation (\ref{rotazione2}) is the parametric form of the parabolic
pseudosphere, seen as surface of revolution about the $z$ axis of
the curve called \emph{tractrix}, which satisfies the following property:
the length of the tangent from the point where it touches the curve
to the point where it intersects the $z$ axis is constant and equal
to $|\rho|$; $\varphi$ is the angle that the tangent to the meridian
makes with the $z$ axis. The angle $\varphi$ varies between $0$ and
$\pi$ so that, keeping $\rho=1$, the related parabolic pseudosphere has vertices at
$z=+\infty$ (corresponding to $\varphi=\pi$) and at $z=-\infty$
(corresponding to $\varphi=0$) and rim at $z=0$ (corresponding to
$\varphi=\frac{\pi}{2}$). The curve is continuous
and regular except at the point $z=0$, which is a cusp point. Choosing
$\rho=-1$ and varying $\varphi$ from $0$ to $\pi$, we shall have the
upward vertical $z$ axis positively oriented ($\varphi$ varying from
$0$ to $\frac{\pi}{2}$), and the downward vertical $z$ axis negatively
oriented ($\varphi$ varying from $\frac{\pi}{2}$ to $\pi$).
In accordance with Hilbert's theorem, for which
it is impossible to embed the entire hyperbolic disk onto $\R^3$,
and since we want the immersion from the horocycle to the
pseudosphere to be regular, then the image must be contained either in the downward
component or in the upward component of the pseudosphere; thus it
does not contain the cuspidal rim. Accordingly, taking $\rho=-1$ once and for all,
the first form, written in Beltrami coordinates, reads
\beq
\label{beltramiforma}
I=\rmd u^2+\left(e^{-u}\rmd v\right)^2 ~~~~~~~ (u\geqslant 0),
\eeq
and the tractrix is:
\beq
\label{trattrice3}
\begin{array}{ll}
x=-\sin{\varphi}, \\
z=-\left(\cos{\varphi}+\ln\left|\tan{\frac{\varphi}{2}}\right|\right).
\end{array}
\eeq
Then, in Sec. \ref{subse:local} we use the form (\ref{beltramiforma})
with $u\geqslant 0$, varying $\varphi$ from $\frac{\pi}{2}$ to $0$ and, accordingly,
$z\geqslant 0$ (see Fig. \ref{fig_2}). In Sec. \ref{subse:sine}, where the coordinate $u$ does not enter
the game and we have chosen the coordinates $p$ and $q$, it is convenient
to vary $\varphi$ from $\frac{\pi}{2}$ to $\pi$ and,
in accordance, taking the downward vertical axis negatively oriented, i.e., $z\leqslant 0$.

\end{document}